\documentclass[12pt]{article}
\usepackage{bbm}

\usepackage{times}
\usepackage{latexsym}
\usepackage{amsfonts,amsthm,amssymb}
\usepackage{euscript}
\usepackage{amstext}
\usepackage{graphicx}
\usepackage{color}
\usepackage{url,hyperref}

\usepackage{subfigure}

\newtheorem{lemma}{Lemma}
\newtheorem{theorem}[lemma]{Theorem}

\newtheorem{definition}[lemma]{Definition}

        {\hspace*{\fill}$\Box$\par}

\title{On the Impossibility of Decomposing Binary Matroids }

\author{Marilena Leichter\thanks{Department of Mathematics, Technical University of Munich, Germany.  Supported in part by the Alexander von Humboldt Foundation with funds from the German Federal Ministry of Education and Research (BMBF) and by the Deutsche Forschungsgemeinschaft (DFG), GRK 2201.} 
\and
Benjamin Moseley \thanks{Tepper School of Business, Carnegie Mellon University, USA. Supported in part by a Google Research Award, an Infor Research Award, a Carnegie Bosch Junior Faculty Chair and NSF grants CCF-1824303,  CCF-1845146, CCF-1733873 and CMMI-1938909.}
\and Kirk Pruhs\thanks{Computer Science Department, University of Pittsburgh, USA. Supported in part  by NSF grants  CCF-1535755, CCF-1907673,  CCF-2036077 and an IBM Faculty Award.}
}

\usepackage{float}

\newcommand{\cI}{{\cal I}}

\newcommand{\cn}{k} 
\newcommand{\gam}{M} 

\providecommand{\keywords}[1]
{
  \small	
  \textbf{\textit{Keywords---}} #1
}

\begin{document}

\maketitle

\begin{abstract}
We show that there exist $k$-colorable matroids that are not $(b,c)$-decomposable when $b$ and $c$ are constants. A matroid is $(b,c)$-decomposable, if its ground set of elements can be partitioned into sets $X_1, X_2, \ldots, X_l$ with the following two properties. Each set $X_i$
has size at most $ck$.  Moreover, for all sets $Y$ such that $|Y \cap X_i| \leq 1$ it is the case that $Y$ is $b$-colorable. A $(b,c)$-decomposition is a strict generalization of a partition decomposition and, thus, our result refutes a conjecture from \cite{berczi2019list}.
\end{abstract}


\section{Introduction}

Consider a matroid $M= (S, \mathcal{I})$ where $S$ is the ground set of elements and $\mathcal{I}$ is the collection of independent sets.  $M$ is said to be $k$-colorable if $S$ can be partitioned in $k$ sets $C_1, C_2, \ldots, C_k$ such that $C_i \in \mathcal{I}$ for all $i \in [k]$.  The smallest $k$ for which $M$ is $k$-colorable is known as the coloring number of the matroid $M$. An optimal coloring of a matroid can be computed in polynomial time \cite{Edmonds65}. 
This is not necessarily the case anymore if we consider, instead of a single matroid, the intersection of $h$ matroids. Consider a collection of $h$ matroids on the same ground set $M_i = (S, \mathcal{I}_i)$ for $i \in [h]$.  The intersection of $M_1, M_2, \ldots, M_h$ is said to be $k$-colorable if $S$ can be partitioned in $k$ sets $X_1, X_2, \ldots X_k$ such that $X_j \in \bigcap_{i=1}^h \mathcal{I}_i$ for all $j$.  That is, each $X_j$ is independent in all of the $h$ matroids. The coloring number of the intersection of $M_1, M_2, \ldots, M_h$ is the smallest $k$ for which the given intersection is $k$-colorable. Matroid intersection coloring is known to be NP-hard for $h\geq 3$ \cite{OBSZARSKI201748}.


\cite{im2020matroid} showed that if each of the $k$-colorable matroids $M_1, \ldots, M_h$ is $(b, c)$-decomposable, the intersection of these matroids can be colored with $k \cdot h \cdot c \cdot b^h$
colors.
\begin{definition}[$(b,c)$-decomposable]

 A $k$-colorable matroid $\gam = (S, \mathcal{I})$ is $(b,c)$-decomposable if there is a partition  $X= \{X_1, X_2,\dots, X_\ell\}$ of $S$
     such that:
     \begin{itemize}
     \item  For all $i \in [\ell]$, it is the case that $|X_i| \leq c \cdot \cn$, and
     \item every set $Y =\{v_1, \ldots, v_{\ell}\}$, such that $v_i \in X_i$,  is $b$-colorable. 
     \end{itemize}
We refer to $X$ as a $(b,c)$-decomposition.
\end{definition}
If $b=1$ then $X = \{X_1, X_2,\dots,X_\ell \}$ represents a partition matroid,
and thus \cite{berczi2019list} called the $(1, c)$-decomposition a partition reduction.
Furthermore, \cite{im2020matroid} showed that if the $(b, c)$-partitions are given for a collection of matroids on the same ground set, or can be efficiently computed,
then the coloring of their intersection can be efficiently computed. 
Note that if  $h$, $b$ and $c$ are all $O(1)$ then
the resulting coloring is an $O(1)$-approximation to an optimal coloring as 
the coloring number for each individual matroid lower bounds the coloring number for
the intersection.

Furthermore, \cite{berczi2019list,im2020matroid,MarilenaGammoid} showed that many common types of 
matroids, including transversal matroids, laminar matroids, graphic matroids and gammoids,
have $(1, 2)$-decompositions. Moreover, they showed that these decompositions can be computed efficiently from the standard representations of these matroids. Thus \cite{berczi2019list} reasonably conjectured 
that every  matroid is $(1, 2)$-decomposable. If this conjecture held, and such decompositions could be found efficiently, then the result from \cite{im2020matroid} would yield an efficient $O(1)$-approximation algorithm for coloring the intersection of $O(1)$ arbitrary matroids.

This paper's main result is that there are matroids that are not $(O(1), O(1))$-decomposable. This refutes the conjecture from \cite{berczi2019list}. In particular, we show that the binary matroid, consisting of the $2^n -1 $ nonzero vectors of dimension $n$, is not  $(O(1), O(1))$-decomposable.

Before proving our main result in Section \ref{sect:main}, we review related work and basic definitions. 

 \subsection{Other Related  Work} 
 \cite{Aharoni06theintersection} showed that for two matroids $M_1$ and
$M_2$, with coloring numbers $k_1$ and $k_2$, the coloring number $k$ of $M_1 \cap M_2$ is at most $2 \max(k_1, k_2)$.
The proof in \cite{Aharoni06theintersection}  uses topological arguments  
that do  not  directly  give  an efficient algorithm  for  finding the coloring. 
\cite{berczi2019list} also showed how to use the existence of $(1, c)$-decompositions
to prove the existence of certain list colorings. 

Motivated by applications to the matroid secretary problem,
 \cite{Karlin} independently showed
 that the same binary matroid that we consider
is not $(1, O(1))$-decomposable. 

\subsection{Definitions}

A \emph{hereditary set system} is a pair $\gam = (S,\cI)$ where  $S$ is a universe of $n$ elements and $\cI\subseteq 2^S$ is a collection of subsets of $S$
with the property that if
  $A \subseteq B \subseteq S$ and $B \in \cI$ then $A \in \cI$.
 The sets in $\cI$ are called \emph{independent}. 
A subset $R$ of $S$ is \emph{$k$-colorable} if $R$ can be partitioned into $k$ independent sets. 
The \emph{coloring number} of  $\gam$ is the smallest $k$ such that $S$ is
 $k$-colorable.
   The \emph{rank} $r(X)$ of a subset $X$ of $S$ is the maximum cardinality of an
   independent subset of $X$.
  A \emph{matroid} is an hereditary set system with the additional properties that $\emptyset \in \mathcal{I}$ and if 
   $A  \in \cI $, $B \in \cI$, and $|A |  < |B|$ then
   there exists an $s \in B \setminus A$ such that $A \cup \{s\} \in \cI$.
   The intersection of matroids $(S, \cI_1), \ldots ,(S, \cI_h)$
   is a hereditary set system with universe $S$ where a set $I \subseteq S$ is 
   independent if and only if for all $i \in [1, h]$ it is the case
   that $I \in \cI_i$. 
    A \emph{flat} $F$ of $M$ is subset of $S$ such
   that for all elements $y \in S\setminus F$ it is the case that adding $y$ to
   $F$ strictly increases the rank.

\section{Main Result: Binary Matroids are Not Decomposable}
\label{sect:main}

This section focuses on showing that binary matroids are not $(b,c)$-decomposable for constants $b$ and $c$.

\begin{definition}
Let $\gam = (S,\cI)$ be the binary matroid where $S$ consists of all $n$ dimensional vectors
with entries that are either 0 or 1, with the exception of the all zero vector. 
 A subset $R$ of $S$ is independent if and only if
the elements of $R$ are linearly independent over the field 
with the elements 0 and 1 with addition and multiplication modulo 2. 
\end{definition}
 
Note that $S$ contains $2^n -1 $ elements and has rank $n$. 

\begin{lemma}
\label{lem:colorflat}
The coloring number of any rank $d$ flat of $\gam$ is  $\lceil  (2^d - 1)/d \rceil$.
Thus, by taking $d=n$, the
coloring number $k$ of $\gam$ is precisely $\lceil  2^n/n \rceil$.
\end{lemma}
\begin{proof}
It is well known that a matroid can be colored with $k$ colors if and only if for
every subset $R$ of elements, $k \cdot r(R) \ge |R|$, that is, $k$ times the rank of $R$ is at least the cardinality of $R$~\cite{Edmonds65}. 
The maximum value of  $|R|/r(R)$ over subsets $R$ of a rank $d$ flat $F$   occurs when $R = F$. Thus this maximum is $(2^d - 1)/d$. 
\end{proof}

\begin{lemma}
\label{lem:flatcount}
If $d \le n/2$ then the number of distinct rank $d$ flats of $\gam$ is at least
$ \frac{2^{dn}}{2^{d^2+d}}$.  
\end{lemma}
\begin{proof}
Consider the process of picking one by one a collection of $d$ vectors to form a basis of 
a rank $d$ flat $F$.
When considering the $i$th choice, there are $(2^n -1 ) - (2^{i-1} -1 )$ choices of elements of 
$S$ that are linearly independent from the previous choices. As the order of the
$d$ vectors chosen  does not matter, the number possible collections
of elements that form a basis of rank $d$ flat is the following.  
$$\frac{\prod_{i=1}^d \left((2^n -1 ) - (2^{i-1} -1 )\right)}{d!}$$
Similarly for a particular rank $d$ flat $F$ there are 
$$\frac{\prod_{i=1}^d \left((2^d -1 ) - (2^{i-1} -1 )\right)}{d!}$$
collections of elements from $F$ that form a basis for $F$.
Thus there are 
$$\frac{\prod_{i=1}^d \left((2^n -1 ) - (2^{i-1} -1 )\right)}{\prod_{i=1}^d \left((2^d -1 ) - (2^{i-1} -1 )\right)} = \prod_{i=1}^d \left(\frac{2^n - 2^{i-1}}{2^d -2^{i-1}}\right)$$
flats of rank $d$.
Lower bounding each term 
in the product in the numerator 
by $2^n - 2^d$, and upper bounding each term in the product
in the denominator by $2^d$, we can conclude that
there are at least
$$\left(\frac{2^n - 2^d}{2^d} \right)^d $$
flats of rank $d$. Then if 
$d \le n/2$, this is at least  
$\frac{2^{dn}}{2^{d^2+d}}$. 
\end{proof}

\begin{theorem}  
If $\gam$ admits a $(b, c)$-decomposition then it must be the case that $4 c^2 2^{d^2 + d} \ge n$, where $d$ is the minimum integer
such that $(2^d -1)/d > b$. In particular, 
for sufficiently large $n$, $M$ admits
no $(O(1), O(1))$-decomposition.
 \end{theorem}

\begin{proof}
Consider an arbitrary $(b, c)$-decomposition 
$X = \{X_1, X_2,\dots,X_\ell \}$ of $\gam$.
As $(2^d -1)/d > b$, a flat of rank $d$ is not $b$-colorable by Lemma \ref{lem:colorflat} . 
Thus for each rank $d$ flat $F$, 
at least two elements of $F$ must be in the same part in $X$. Otherwise, we get a contradiction to
the definition of $(b, c)$-decomposability. To see this, consider setting $Y$ to $F$ in the definition of the $(b,c)$-decomposition.  That is, each element of $F$ is selected to be in $Y$ as this includes at most one element in any part $X_i$. The resulting representatives would
not  be $b$-colorable by the above characterization of $F$. 
If two elements of a  rank $d$ flat $F$ are in the same part $X_i \in X$ then
we say that $F$ is covered by part $X_i$.

Since $X$ is a $(b, c)$-decomposition, the cardinality
of each part of $X$ is at most $ck$. 
Each pair of elements  $x, y$ in a part  $ X_i \in X$ can be contained in at
most ${2^n \choose d-2}$ rank $d$ flats. To see  this note that each rank $d$ flat $F$ can be represented by $d$ independent basis vectors in $F$, and since $x$ and $y$ are already specified, there are at most $d-2$ more choices for these basis vectors.
There are at most ${ck \choose 2}$ possible pairs of elements 
from a part  $X_i \in X$, and $X_i$ can cover at most ${ck \choose 2} {2^n \choose d-2}$ different flats.
Thus in aggregate, all the parts of $X$ can cover at most 
$\ell {ck \choose 2} {2^n \choose d-2}$ flats.
Then using the fact that $\ell$ is at most $  n $, 
$k $ is at most $2 \cdot 2^n / n$, and upper bounding
${x \choose y}$ by $x^y$, 
we can conclude that in aggregate all the parts of $X$ can cover at most 
$\ell {ck \choose 2} {2^n \choose d-2} \leq n (ck)^2 (2^n)^{d-2}\leq 4 c^2 2^{n d}/n$   flats. 
Since each of the flats must be covered by some part of $X$, and since 
by Lemma \ref{lem:flatcount} the number of rank $d$ flats is at least
$\frac{2^{nd}}{2^{d^2+d}}$, 
it must be the case that
$$4c^2 2^{n d}/n \ge \frac{2^{nd}}{2^{d^2+d}}$$
or equivalently $4 c^2 2^{d^2 + d} \ge n$.
\end{proof}

\bigskip
\noindent
{\bf Acknowledgements:} We thank James Oxley for helpful discussions. 

\bibliographystyle{alpha}

\bibliography{bibgammoid}

\end{document}